\documentclass[10pt, conference]{IEEEtran}
\IEEEoverridecommandlockouts
\usepackage{algorithm}
\usepackage{algpseudocode}
\usepackage{amsmath,amssymb,amsfonts,amsthm}
\usepackage{graphicx}
\usepackage{textcomp}
\usepackage{color}
\usepackage{colortbl}
\usepackage{booktabs}
\usepackage{bbm}
\usepackage{subfigure}
\usepackage{url}
\usepackage{multirow}
\newtheorem{definition}{Definition}
\newtheorem{theorem}{Theorem}
\usepackage{balance}
\usepackage{caption}
\usepackage[numbers,sort&compress]{natbib}
\usepackage{tikz}
\newcommand{\wjs}{\textcolor{black}}
\newcommand*\circled[1]{\tikz[baseline=(char.base)]{%
               \node[shape=circle,fill=white!10,draw,inner sep=0.5pt] (char) {#1};}}

\usepackage{enumitem}
\makeatletter
\DeclareRobustCommand{\bxz}{\ifmmode \mathbxz
  \else
    \leavevmode\unskip\penalty9999 \hbox{}\nobreak\hfill
    \quad\hbox{\bxzsymbol}%
  \fi
}

\DeclareRobustCommand{\bxz}{\ifmmode \mathbxz
  \else
    \leavevmode\unskip\penalty9999 \hbox{}\nobreak\hfill
    \quad\hbox{\bxzsymbol}%
  \fi
}
\newcommand{\mathbxz}{\quad\hbox{\bxzsymbol}}

\providecommand{\bxzsymbol}{\fbox{\footnotesize B.X.Z}}
\providecommand{\fooname}{Foo}
\usepackage[utf8]{inputenc}

\begin{document}

\title{FedServing: A Federated Prediction Serving Framework Based on Incentive Mechanism}

\author{\IEEEauthorblockN{Jiasi Weng${}^*$${}^\dag$\thanks{${}^*$Work was done when the author was a visiting student with City University of Hong Kong. It has been accepted for inclusion in IEEE INFOCOM 2021.}, Jian Weng${}^*$, Hongwei Huang${}^*$, Chengjun Cai${}^\dag$, and Cong Wang${}^\dag$}
\IEEEauthorblockA{${}^*$Jinan University;
${}^\dag$City University of Hong Kong
}}

\maketitle

\begin{abstract}
%
Data holders, such as mobile apps, hospitals and banks, are capable of training machine learning (ML) models and enjoy many intelligence services. To benefit more individuals lacking data and models, a convenient approach is needed which enables the trained models from various sources for prediction serving, but it has yet to truly take off considering three issues: (\emph{i}) incentivizing prediction truthfulness; (\emph{ii}) boosting prediction accuracy; (\emph{iii}) protecting model privacy.

We design FedServing, a federated prediction serving framework, achieving the three issues. First, we customize an incentive mechanism based on \emph{Bayesian game theory} which ensures that joining providers at a \emph{Bayesian  Nash  Equilibrium} will provide truthful (not meaningless) predictions. Second, working jointly with the incentive mechanism, we employ \emph{truth discovery algorithms} to aggregate truthful but possibly inaccurate predictions for boosting prediction accuracy. Third, providers can locally deploy their models and their predictions are securely aggregated inside TEEs. Attractively, \emph{our design supports popular prediction formats, including top-1 label, ranked labels and posterior probability.} Besides, blockchain is employed as a complementary component to enforce exchange fairness. By conducting extensive experiments, we validate the expected properties of our design. We also  empirically  demonstrate  that  FedServing reduces the risk of certain \emph{membership inference attack}.\looseness=-1
\looseness=-1
\end{abstract}


\section{Introduction}
Machine learning (ML) is revolutionizing our world and the global market for ML driven services is expected to reach \$$5,330$ million by 2024~\cite{ml2019market}.
%
Many data holders, such as mobile apps, hospitals and banks, are able to train models based on the available data they hold, and use the trained models to achieve functionality and business innovation~\cite{bello2016social}.
%
%
%
From another perspective, most individuals lacking data and power are incapable of training models, so that they hardly benefit from ML.
Even if an individual is in possession of a model, it still has the real-world demand to collaborate with others' models, demonstrated by an existing real-world case, \emph{i.e.}, two banks in North America collaborate to detect money laundering.
Obviously, due to privacy concerns, intellectual property issues or business competition, model owners are unwilling to share their trained models.
Thus, it is necessary to build a bridge which connects model owners who have no incentives of sharing models with individuals who need models.\looseness=-1

%
%
%
Building such a bridge inevitably needs to support three essential requirements as following:
(\emph{i}) providing sufficient \emph{incentives} to the model owners so that they are willing to contribute their models;
(\emph{ii}) enabling individual users of interest to enjoy as \emph{high-performance} as possible models;
(\emph{iii}) guaranteeing \emph{model privacy}, since models imply private information about their training data~\cite{song2017machine}.
However, there exists no work to realize such a bridge, so that it has yet to truly take off.\looseness=-1

While Machine-Learning-as-a-Service (MLaaS) platforms enable monetizing models for prediction serving on a pay-per-query basis, trained models have to reside on the untrusted servers, causing model privacy concerns.
%
Although earlier works present effective approaches~\cite{juvekar2018gazelle, mireshghallah2020shredder, tramer2018slalom} for protecting models against the untrusted servers, they still are not really satisfactory.
Specifically, cryptographic methods are computation-consuming and inefficient when handling large-sized models~\cite{juvekar2018gazelle}, but high-performance models usually are large.
%
%
%
Differential privacy based defenses would sacrifice prediction accuracy~\cite{mireshghallah2020shredder}.
Trusted hardware-enabled approaches are relatively practical but still have efficiency limitation.
It is due to that trusted hardwares are majorly restricted to CPUs, but running large models usually needs GPUs~\cite{tramer2018slalom}.
Motivated by our observations, our goal is to make model owners freely deploy their models without limits, collectively contribute their models to make profits and securely use models without privacy leakage concerns.
Towards the goal, we present a \emph{federated prediction serving framework, FedServing,} towards model owners from various sources in an open setting.
Our starting point is allowing model owners to deploy models at local devices and provide aggregated predictions for exchanging with monetary rewards.
Standing on top of it, we especially make efforts to design our solution for enforcing prediction accuracy, due to the following two-fold challenges:
%
%
%

%

\noindent\textbf{Challenge (\emph{i}): Strategic behaviors of model owners.}
Model owners (hereafter called as providers) are likely to be rational and selfish, so that they may be strategic to report meaningless predictions without effort.
In addition, ground truths with respect to given prediction queries are usually unknown, which makes the truthfulness of predictions hard to be verified.

\noindent\textbf{Challenge (\emph{ii}): Varying quality of models.}
%
%
While aggregating predictions (\emph{e.g.}, via majority voting or averaging) is a classic strategy for improving accuracy, they may be less effective in our case.
The issue of majority voting and averaging is that they assume prediction sources (\emph{i.e.}, models) are equally reliable.
Yet, the assumption cannot hold in our open setting.
It is due to that
(1) the qualities of models from various sources are varying and due to local deployment, there is no available authority enforcing the model quality upon answering prediction queries;
(2) even well-trained models are not always generalized well over the whole feature space of all prediction queries, so producing predictions are probably not always accurate.\looseness=-1

In light of the two challenging issues, the state-of-the-art solutions  usually resort to incentive mechanisms in conjunction with quality-aware aggregation algorithms, as learnt from the literature~\cite{peng2015pay, yang2017designing, jin2017theseus, sun2020towards, gong2018incentivizing, zhao2020pace}.
Unfortunately, previous work cannot be used to mitigate our challenging issues.
The main reason is that they fail to simultaneously handle categorical and continuous data covered by popular prediction outputs~\cite{mangai2010survey}.
As a concrete instance, the prediction outputs for each query of a sentiment analysis task can be \textbf{\emph{top-1 label}}, \emph{e.g.}, \textsf{[upset]}, \textbf{\emph{ranked labels}}, \emph{e.g.}, \textsf{[upset, scared, distressed, guilty]}, and \textbf{\emph{posterior probability}} for each label, \emph{e.g.}, $\left \langle 95.0\%, 2.0\%, 1.0\%, 2.0\% \right \rangle$, which are supported by Google Photos and Google Cloud Vision API, for example.\looseness=-1

\textbf{Our key design.} We customize a complementary mechanism by integrating an incentive design with ``truth-finding" algorithms.
Concretely, our mechanism
(1) uses \emph{Bayesian game theory} to model the honest and strategic behaviors of providers and ensures the existence of a \emph{Bayesian Nash Equilibrium}, where all providers will offer truthful (rather than meaningless) predictions for given prediction queries;
(2) employs \emph{truth discovery (TD) algorithms} to learn highly accurate predictions from the truthful (but possibly inaccurate) predictions to eliminate the effect of inaccurate predictions;
(3) allocates the providers with fair rewards in proportion to the truthfulness of their predictions;
(4) simultaneously handles prediction output formats including labels and the respective posterior probability.

Despite that models are locally deployed, privacy concerns still exist due to disclosing predictions of a model.
Concretely, a model's predictions can be exploited to infer if a data record was used to train a model, \emph{e.g.}, identifying if an individual was a patient at the hospital, known as membership inference attacks~\cite{shokri2017membership}.
%
%
%
To address the privacy concern, we leverage trusted execution environments (TEEs) to aggregate predictions from multiple providers, and only aggregated predictions are revealed to users~\cite{papernot2018marauder}.
Owing to the confidentiality and integrity provided by TEEs, a model's predictions are not revealed and aggregated predictions are correctly generated.
It is noteworthy that our proposed incentive mechanism also can benefit from the TEEs' integrity, since the procedure of evaluating the truthfulness of predictions from each model can be correctly executed, which further enforces fair rewards guided by the truthfulness.
Notably, we do not use privacy-preserving verifiable cryptography, considering that TEEs are relatively more performant.

Besides, we need to facilitate an open setting for model owners from various sources freely joining in FedServing.
But meanwhile, we also need a regulation complementary to our incentive mechanism for fulfilling the transparent process of money settlement and deterring providers' selfish behaviors, \textit{e.g.}, abortion, thereby achieving the fairness of money-prediction exchange among users and providers.
In light of the issues, we choose blockchain to facilitate the open setting and enforce the regulation.\looseness=-1

We note that FedServing can be extended to support the existing prediction serving systems and now we shed light on the service manner of our FedServing framework.
A prediction serving server can deploy a smart contract as a uniform query interface for charging users and as an entrance for participating providers.
When receiving the query and fees from a user, the server resorts to its off-chain TEEs-empowered component to collect predictions from participating providers who undertake the prediction task.
The TEE strategically aggregates predictions and submits aggregated predictions to the blockchain.
Finally, the user obtains the predictions and meanwhile the smart contract allocates the user's fees to the participating providers according to the truthfulness of their predictions.

In conclusion, this paper makes the main contributions as following:
\begin{itemize}
    \item \wjs{We propose a federated prediction serving framework empowered by the blockchain, providing an as accurate as possible prediction service with truthful contributions from various source models in an open setting.}
    \item \wjs{We customize an incentive mechanism for eliciting truthful contributions, by carefully applying a technique of \emph{peer prediction}~\cite{radanovic2014incentives} and fully respecting the formats of popular prediction outputs.}
    \item \wjs{We extend a widely-adopted truthful discovery algorithm to support our prediction setting, and make it jointly work with our designed incentive mechanism, finally producing as accurate as possible predictions.}
    \item \wjs{We implement our design and conduct extensive experiments in terms of the performance, validity and ability against a privacy attack. For reproducibility, our code is publicly available at \url{https://github.com/H-W-Huang/FedServing}.}
\end{itemize}
%



\section{Related Work}\label{sec:related}
\noindent \textbf{Prediction Serving System.}
Existing excellent systems~\cite{crankshaw2017clipper, lee2018pretzel} centralizedly manage models as well as deploy models for low-latency and high-throughput prediction serving, where models are off-the-shelf.
%
%
To enhance prediction accuracy, they generally support ensemble models which aggregate predictions from multiple models.
%
%

Different from them, our work focuses on the models from various sources for prediction serving in an open setting. More precisely, we consider \emph{how to incentivize model owners from various sources to provide truthful prediction services while respecting model privacy and ensuring prediction accuracy.}
To the end, we present a distributed framework which achieves the following three-fold components which are less considered by the existing systems~\cite{crankshaw2017clipper, lee2018pretzel}.
%
%
%
%
%

\noindent(\emph{i}) \textbf{Pricing mechanism.} We customize a pricing mechanism for compensating participating providers and incentivizing prediction truthfulness, instead of using an one-price-fits-all pricing structure, which still respects the pay-per-query business pattern of the current MLaaS platforms.

\noindent(\emph{ii}) \textbf{Quality-aware aggregation.} Considering that the model quality and the ground truths of prediction queries are unknown in our open setting, we use TD algorithms to aggregate predictions rather than simply averaging, thereby eliminating the effect of low-accuracy predictions.

\noindent(\emph{iii}) \textbf{Model and prediction protection.} We make models never leave local devices and multiple predictions are securely aggregated inside TEEs so that users only obtain aggregated predictions.
Due to local deployment, providers retain control over when and how their in-house models are used to make predictions, \textit{e.g.}, joining in an ensemble to produce predictions, thereby reducing the risks of privacy attacks~\cite{shokri2017membership}.
%

%
%
%

%
%
%


\noindent \textbf{Incentive Mechanism.}
%
Prior incentive mechanisms~\cite{zhang2012reputation, luo2014profit, jin2016enabling, jin2016inception, jin2017centurion, jin2017theseus, gong2018incentivizing, jin2019dynamic, zhang2015incentivize, wang2017melody, zhang2015truthful, yang2012crowdsourcing, han2016posted, peng2015pay, zhao2014crowdsource, zhang2014free, chen2016incentivizing, huang2019crowdsourcing, zhang2016incentive, yang2017designing, sun2020towards, zhao2020pace} are designed for
stimulating participation by compensating workers' costs with monetary rewards, and implementing the economic properties, such as platform profit maximization, individual rationality and budget feasibility, which greatly promote the development of crowdsourcing.
%
The incentive mechanisms generally resort to game-theoretic methods,
such as reverse auction~\cite{zhang2014free, jin2016enabling, jin2016inception, zhang2015incentivize, wang2017melody}, double auction~\cite{jin2017centurion, zhang2016incentive} and all-pay auction~\cite{luo2014profit}, or other game theory~\cite{jin2017theseus, gong2018incentivizing, huang2019crowdsourcing}.
%
With the game-theoretic analysis, they consider the strategic behaviors of workers and investigate how to encourage workers to behave truthfully.\looseness=-1

In this paper, we aim to stimulate the truthfulness of collective predictions considering providers' strategic behaviors of providing meaningless predictions, thereby achieving quality control.
Existing mechanisms do not solve our problem, since the following four \emph{requirements} cannot be simultaneously satisfied:\looseness=-1

\noindent(\emph{i}) \textbf{Incentivizing truthfulness.}
Most mechanisms~\cite{jin2016enabling, jin2016inception, zhang2015incentivize, wang2017melody} focus on incentivizing workers to reveal their costs truthfully.
A few excellent mechanisms like~\cite{jin2017theseus, gong2018incentivizing} incentivize the truthfulness of crowd data as this paper, but they are unsatisfactory to us due to that (\emph{ii}) below cannot be supported.

\noindent(\emph{ii}) \textbf{Simultaneously handling categorical and continuous data.}
Theseus~\cite{jin2017theseus} proposes a truthful mechanism for quality and efforts elicitation while focusing on continuous sensing data.
\cite{gong2018incentivizing} creatively studies the joint elicitation of quality, efforts and data while focusing on discrete data (more precisely, binary data).
Their techniques do not solve our problem, since we simultaneously consider labels and confidence values which are categorical data and continuous data, respectively.

\noindent(\emph{iii}) \textbf{No reliance on prior knowledge.}
Prior arts~\cite{jin2017centurion, yang2017designing} assume the prior knowledge about workers' reliability or reputation that helps allocate rewards. We do not use such prior knowledge, since a provider's predictions for historical tasks are considered irrelevant to the current task\footnote{We note that a prediction task usually is requested with batch queries, where the queries belonging to the same task are relevant.}.

\noindent(\emph{iv}) \textbf{Jointly addressing incentive and quality concerns.}
Most mechanisms for incentivizing truthfulness do not jointly work with TD, except to~\cite{jin2017theseus, yang2017designing, peng2015pay, zhao2020pace, sun2020towards}.
These works can be deployed with TD, but they still are unsatisfactory to us:
Theseus~\cite{jin2017theseus}, \cite{yang2017designing} and \cite{zhao2020pace} majorly consider the continuous data stream;
\cite{peng2015pay} does not focus on workers' strategic behaviors; \cite{sun2020towards} cares about binary answers and assumes that most workers are reliable.
%

%
%

\section{System Overview}\label{sec:overview}
In this section, we present our FedServing framework. It begins with the system model, and then figures out the threat assumptions and design goals.
\begin{figure}[h]
\small
\vspace{-5pt}
    \centering
    \includegraphics[width=0.7\columnwidth]{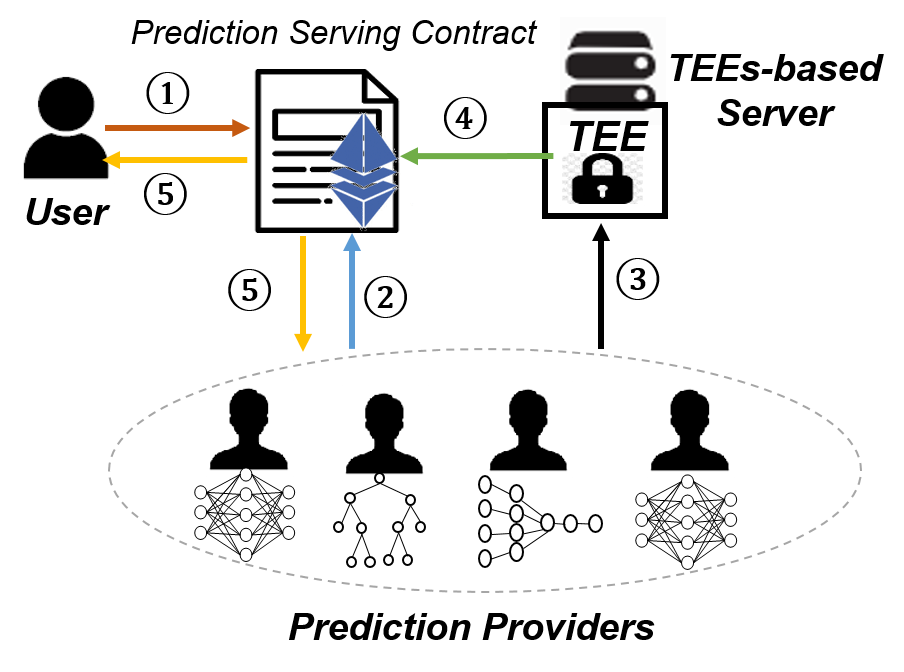}
    \caption{Overview of the FedServing framework.}
    \label{fig:overview}
    \vspace{-10pt}
\end{figure}
\subsection{System Model}\label{system_model}
At the high level, our FedServing consists of four entities as shown in Fig.~\ref{fig:overview}: \textit{prediction providers, user, smart contract} and \textit{TEEs-based server}.
Specifically, the prediction providers who own various ML models monetize their prediction query services on the blockchain (\textit{e.g.}, Ethereum).
They could publicize non-private model profiles like service APIs for user accessing their models.
User is able to browse model profiles on the blockchain, and query prediction services via a smart contract, named as \textit{prediction serving contract} (PS contract).
PS contract aims at receiving the user's query request, relaying the request, receiving aggregated predictions and achieving the fair payment finalization.
As the intermediator between the prediction providers and the PS contract, the TEEs-based server is responsible for strategically aggregating predictions sent by multiple providers, and calculating accuracy-aware scores used to guide allocating rewards.
%
%
The basic workflow in Fig.~\ref{fig:overview} is described as following:\looseness=-1
\begin{enumerate}[label=\protect\circled{\arabic*}]
    \item User sends a transaction which contains a description about her requested task, \textit{e.g.}, a sentiment analysis task, and makes a deposit for payment to the PS contract.
    Note that the task's input data, \textit{e.g.}, text files, can be securely stored in an accessible system like IPFS, and then be securely authorized to participating providers.
    %
    %

    %
    \item Providers participate in the task by submitting a deposit to the PS contract for potential penalty, \emph{e.g.}, punishing abortion.
    Here, we omit the phase that they can authentically obtain the task's input data from IPFS.
    \item Participating providers evaluate local models on the input data, and lastly submit predictions to the TEE via an authenticated communication channel.
    \item The TEE strategically aggregates predictions from multiple participating providers and compute accuracy-aware scores for each provider. After that, the aggregated predictions are correctly encrypted using the user's public key and submitted to the blockchain.
    \item User retrieves and decrypts the aggregated predictions using her private key, and meanwhile, her deposit is allocated to the participating providers according to the respective accuracy-aware scores.\looseness=-1
\end{enumerate}
%
%
%
%

%
%
%

\subsection{Threat Model and Assumptions}\label{sec:assumption}
\noindent\textbf{Prediction Provider.} We consider that prediction providers are rational and self-interested.
They may act to maximize their profits by submitting arbitrary predictions.
%
%
%
The providers answering certain query are named as participating providers and assumed not to collude with each others.
In addition, we assume that the input data received by participating providers are benign; perturbed input data known as \textit{adversarial examples}~\cite{szegedy2013intriguing} are out of our consideration.

\noindent \textbf{TEEs.}
%
We trust that TEEs, \textit{e.g.}, Intel Software Guard Extensions (SGX), can securely execute specific programs against external observation and manipulation, \textit{i.e.}, ensuring confidentiality and integrity.
We note that side-channel attacks and rollback attacks on TEEs are out of the scope of this paper like prior TEEs-empowered work~\cite{ hunt2018ryoan}, owing to many off-the-shelf defence mechanisms~\cite{ahmad2018obliviate, kaptchuk2019giving}.
We rely on the authenticated communication channels built between a TEE and a remote party, \textit{e.g.}, Intel SGX's Enhanced Privacy ID (EPID) remote attestation protocol.\looseness=-1

\noindent \textbf{Blockchain.} We trust the blockchain for integrity and availability.
Smart contract autonomously and faithfully executes defined functions, \textit{e.g.}, correctly locking deposits and settling rewards, which is assumed not vulnerable to software bugs.

\noindent\emph{Remarks.}
We aware that FedSeving can suffer from Sybil attacks~\cite{douceur2002sybil}, where a prediction provider may maliciously use multiple fake accounts to join in certain task.
For demoralizing Sybil attacks, a widely adopted solution is to increase the attack cost like solving proof-of-work puzzles and making deposits.
In this paper, we require each participating provider to make a deposit before undertaking a task.

\subsection{Design Goals}\label{sec:goals}
\noindent\textbf{Truthfulness and accuracy.}
It means that user can obtain aggregated predictions with truthfulness and accuracy guarantees.
Specifically, each participating provider provides truthful (but possibly inaccurate) predictions, and meanwhile, the truth discovery algorithm is correctly conducted on the provided truthful predictions to produce truths, \textit{i.e.}, aggregated predictions, which are regarded accurate enough.
%

\noindent\textbf{Fairness.}
It includes the fairness of reward allocation and the fairness of money-prediction exchange.
First, each participating provider in a task gets a fair reward guided by a \textit{strictly proper score} which is computed based on the truthfulness of their predictions.
A comparatively truthful prediction leads to a higher score, and the prediction's provider obtains comparatively more rewards.
%
%
%
Second, all participating providers receive rewards \textit{iff} the user obtains the final predictions.
\begin{table*}[htbp]
\vspace{-10pt}
 \centering
 \caption{Examples of prediction formats}\label{tab:format}
 \begin{tabular}{lccc}
  \toprule
  Format & Model1 & Model2 & Model3 \\
  \midrule
  Abstract & $\left \langle 0, 1, 0, 0 \right \rangle$
       & $\left \langle 1, 0, 0, 0 \right \rangle$
       & $\left \langle 1, 0, 0, 0 \right \rangle$   \\
  Rank & $\left \langle 2, 4, 1, 3 \right \rangle$
       & $\left \langle 4, 2, 1, 3 \right \rangle$
       & $\left \langle 4, 2, 3, 1 \right \rangle$   \\
  Measurement & $\left \langle 2.0\%, 49.0\%, 1.0\%, 48.0\% \right \rangle$
       & $\left \langle 92.0\%, 2.0\%, 1.0\%, 5.0\% \right \rangle$
       & $\left \langle 93.0\%, 2.0\%, 3.0\%, 2.0\% \right \rangle$  \\
  \bottomrule
 \end{tabular}
 \\
 \vspace{2pt}
\vspace{-20pt}
\end{table*}
\section{Design of Prediction Aggregation}\label{sec:combination}
Considering that our FedServing is built in an open setting, participating models might produce inaccurate predictions.
The reasons include that (\emph{i}) varying quality models can freely participate in FedServing, and meanwhile, there is no available authority enforcing the quality of participating models;
 (\emph{ii}) trained models are not always generalized well over the whole feature space of every prediction task~\cite{kuncheva2002switching}.

\wjs{In light of this issue, we study the lessons from the earlier works~\cite{yang2017designing, jin2017theseus, sun2020towards, li2014resolving, su2014generalized} and leverage TD algorithms~\cite{li2014resolving} to aggregate predictions, so as to learn as accurate predictions as possible from varying quality models in absence of ground truth.}
We support three common prediction formats in practice.
To the best of our knowledge, there is no existing scheme dealing with the issue as this paper.
The previous work~\cite{su2014generalized} is similar to our design of prediction aggregation, but it focuses on one single format, \emph{i.e.}, probability vector.
We especially consider other popular prediction outputs, \emph{e.g.}, ranked label list, used in Google Photos.

%
For ease of presentation, we begin with an instance of prediction task.
Then, we elaborate three prediction formats and demonstrate how to aggregate them.

\noindent\textbf{Instance Description.} Suppose that a social psychologist has a sentiment analysis task for a set of consulting letters from anonymous citizens. She needs to label the set of consulting letters with the emotion states for studying social projection. With the task, she can query the PS contract in our FedServing:
\textsf{what are the emotion states for each consulting letter, distressed, upset, guilty or scared?}

\noindent\textbf{Prediction Formats.}
In the above instance, we introduce three popular prediction output formats~\cite{mangai2010survey}:
(1) \textbf{Abstract}: a top-1 class label,  \textit{e.g.}, \textsf{'upset'},
(2) \textbf{Rank}: a ranked list of labels,  \textit{e.g.}, \textsf{[upset, scared, distressed, guilty]}, and
(3) \textbf{Measurement}: a probability vector for possible class labels, \textit{e.g.}, $\left \langle 2.0\%, 95.0\%, 1.0\%, 2.0\% \right \rangle$ for \textsf{[distressed, upset, guilty, scared]} (their sum is $100\%$).

Apparently, the measurement output contains the most detailed prediction information while the abstract output contains less information.
%
%
Note that here we mainly discuss classification tasks, but our method can be easily extended to regression tasks which are associated with real-valued predictions.

\noindent\textbf{Prediction Aggregation.}
We now introduce the algorithm to aggregate predictions adapted to the three formats.
Specifically, in order to fluently run the truth discovery algorithm as shown in Algorithm~\ref{alg:truthdiscovery}, we carefully transform the later two formats into continuous data vectors.
For ease of explanation, we suppose that there are three models predicting a given consulting letter with the corresponding label list \textsf{[distressed, upset, guilty, scared]}.
Their predictions with respect to the three formats are demonstrated in TABLE~\ref{tab:format}.

We now explain how we uniformly represent the three-format predictions by using continuous data vectors.
For the abstract format, the three models separately produce labels \textsf{'upset', 'distressed'} and \textsf{'distressed'}.
We transform them into the corresponding $0/1$ value vectors, where the index with value $1$ is the most possible label, as shown in the abstract row of TABLE~\ref{tab:format}.
For the rank format, the three models provide the ranked lists of possible labels as presented in TABLE~\ref{tab:rank}.
For example, a ranked list \textsf{[upset, scared, distressed, guilty]} is given by the first model.
We set ranked integer values to each ranking level. A largest integer represents the highest ranking level while a smallest integer represents the lowest one.
With this representation rule, the ranked lists in TABLE~\ref{tab:rank} are transformed into the vectors with integer values in the rank row of TABLE~\ref{tab:format}.
Last, the probability vectors in the measurement format are presented without change.
Hereafter, we call the vector values as \emph{confidence values}.

\begin{table}[htbp]
 \centering
 \caption{Examples of ranked lists.}\label{tab:rank}
 \begin{tabular}{cccc}
  \toprule
  Value & Model1 & Model2 & Model3 \\
  \midrule
  4 & \textsf{upset} & \textsf{distressed} & \textsf{distressed} \\
  3 & \textsf{scared} & \textsf{scared} & \textsf{guilty} \\
  2 & \textsf{distressed} & \textsf{upset} & \textsf{upset} \\
  1 & \textsf{guilty} & \textsf{guilty} & \textsf{scared} \\
  \bottomrule
  \vspace{-10pt}
 \end{tabular}

\end{table}

After the uniform representation, multiple predictions for the set of consulting letters in each format will be aggregated via Algorithm~\ref{alg:truthdiscovery} including two steps.
Specifically, we suppose that there are multiple predictions from $m$ $(m \geq 3)$ providers for $n$ consulting letters.
Each prediction is a $c$-length vector containing the confidence values for each class label, where $c$ is the number of given possible class labels.
%
%
They are represented as $\{I^j_i\}_{i=1,j=1}^{m,n}$, where $I^j_i$ is a continuous data vector $\mathbf{v}^j_i=$($v^j_{i1}, ..., v^j_{ic}$).
Now, with Algorithm~\ref{alg:truthdiscovery}, we iteratively estimate the truths on $\{I^j_i\}_{i=1,j=1}^{m,n}$ and update $m$ providers' weights until convergence.
The algorithm finally outputs the truths as the aggregated predictions $\{O^{j(\varepsilon)}\}_{j=1}^n$ with respect to each consulting letter.

\begin{algorithm}[h]
\scriptsize
\caption{Truth discovery}\label{alg:truthdiscovery}
\begin{algorithmic}[1]
\Require
provider predictions $\{I^j_i\}_{i=1,j=1}^{m,n}$
\Ensure
truth predictions $\{O^{j(\varepsilon)}\}_{j=1}^n$

\State Initialize $r=1$ and weights $\{w^{(r)}_i=1\}_{i=1,...,m}$.
\Repeat
%
\For{each $j \in [1,n]$}
\State $O^{j(r+1)} \leftarrow \frac{\sum_{i=1}^m w^{(r)}_i I^j_i}{\sum_{i=1}^m w^{(r)}_i}$ \hspace{86pt}$(1)$
\EndFor
\vspace{1pt}
\For{each $i \in [1,m]$}
%
%
\vspace{1pt}
\State $w^{(r+1)}_i \leftarrow$ $-$log$(\frac{\sum_{j=1}^n f_{loss}(O^{j(r+1)}, I^j_i)}{\sum_{k=1}^m \sum_{j=1}^n f_{loss}(O^{j(r+1)}, I^j_k)})$ \hspace{8pt} $(2)$
\EndFor
\State $r=r+1$
\Until $r \leq \varepsilon $
\\
\Return $\{O^{j(\varepsilon)}\}_{j=1}^n$
\end{algorithmic}
\end{algorithm}

Initially, we set each provider's weight with $1$ and denote an iteration threshold $\varepsilon$.
Then, with fixed weights, $m$ providers' predictions are aggregated via the weighted mean method (Step (1)).
During the iterative computation, the aggregated predictions are closer to that of the providers having higher weights.
With the aggregated predictions, each provider's weight is updated based on the distances between his predictions and the aggregated predictions with respect to $n$ consulting letters (Step (2)).
The provider whose predictions are closer to the aggregated predictions will be assigned with a higher weight.
Here, the loss function $f_{loss}(\cdot)$ is used to characterize the distance and specifically, we use the normalized squared loss function.
Step (1) and (2) are iteratively computed until $r$  reaches pre-defined threshold $\varepsilon$.

%
%
%

\section{Design of Pricing Mechanism}\label{sec:pricing}
The previous section introduces the process of aggregating predictions with the aim to filter out less accurate predictions.
Yet, the accuracy of aggregated predictions still cannot be guaranteed if a majority of self-interested providers offer meaningless predictions.
In order to motivate the self-interested providers to provide truthful predictions, we jointly design our pricing mechanism by employing the Bayesian game theory.
Notably, predictions contain categorical and continuous data which will be simultaneously handled.\looseness=-1

This section begins with the setting definitions and design objectives, and then presents the pricing mechanism formulation and an approximate solution.
To the end, an analysis for the proposed pricing mechanism is elaborated.
%

\subsection{Mechanism Setting}
We use the game theory method to model the strategic behaviors of participating providers inspired by the works~\cite{jin2017theseus, gong2018incentivizing}.
Concretely, we model participating providers $P=\{i, ..., m\}$ playing a \textit{non-cooperative game}, where each of them independently gives a private prediction for each query requested by certain user.
Note that a requested task can include multiple queries, \textit{e.g.}, labeling multiple consulting letters.
%

In the game, participating providers behave as utility maximizers.
They behave strategically by evaluating their expected utility.
Specifically, they will not participate if the expected utility is negative, and otherwise, they offer predictions via a specific strategy that maximizes the expected utility.
In general, the evaluation needs some technical assumptions~\cite{radanovic2014incentives}.
We assume that participating providers undertaking the same task have a common prior belief, and meanwhile, they use the same belief updating procedure, \textit{i.e.}, Bayes' rule.

A provider's behavior is described by \textit{strategy}.
A strategy is denoted by $s=(\mathbf{l}, \mathbf{v})$ meaning giving a prediction for a query , or $\perp$ meaning abort.
Herein, $\mathbf{l}$ is a list of claimed possible class labels and each label in $\mathbf{l}$ is from discrete set $\Omega$;
$\mathbf{v}$ is the corresponding posterior probability values which are drawn from probability density distributions $\Psi$.
Thus, the strategy space is $\{(\Omega, \Psi)\} \cup \{\perp\}$.
Then, the participating providers' strategy profile is $\mathbf{S}=(s_1, ..., s_m)$, if we suppose that there are $m$ participating providers.

Next, we continue to formulate the provider model, the user model and a Bayesian Nash Equilibrium for providers.

\noindent\textbf{Provider Model.}
Within the defined game, a provider's payoff depends on his own strategy with regard to other providers' strategies.
Specifically, given a payment function $p(\cdot)$, a cost function $c(\cdot)$ and deposit $d_0$, we define any provider's utility $u_i(\mathbf{S})$, $i\in P$ in a game with a strategy profile $\mathbf{S}$ as following:
$$
u_i(\mathbf{S}) = p_i(\mathbf{S}) - c(s_i) - d_0.
$$
Next, any provider can evaluate the expected utility:
$$
\mathbb{E}_{\mathbf{S}_{-s_i}}[u_i(s_i,\mathbf{S}_{-s_i})] =
\mathbb{E}_{\mathbf{S}_{-s_i}}[p_i(s_i, \mathbf{S}_{-s_i})] - c(s_i) - d_0,
$$
where $\mathbf{S}_{-s_i}$ is the strategy profile excluding $s_i$.
Note that a participating provider's deposit for $n$ queries is $d=n\times d_0$.

\noindent\textbf{User Model.} A user's objective is to obtain the aggregated predictions whose accuracy is as close as possible to the truth accuracy.
To exchange the aggregated predictions $\{\mathbf{v}^j\}_{j=1, ..., n}$ of $n$ queries from $m$ participating providers, she makes amount of deposits, namely budget $B$, on the blockchain.
Assume that the market publicizes budget curves relative to the number of employed providers via market survey.
%
With the budget curves, the user deposits a budget level that enables soliciting certain number of prediction providers.\looseness=-1

\noindent\textbf{Bayesian Nash Equilibrium.} A strategy profile $\mathbf{S^*}$ is denoted as a Bayesian Nash Equilibrium (BNE) in the defined game, if no provider $i\in P$ can increase her expected utility by changing the current strategy $s^*_i$ with regard to other providers' strategies $\mathbf{S}^*_{-s_i}$:
$$
\mathbb{E}_{\mathbf{S}^*_{-s_i}}[u_i(s^*_i,\mathbf{S}^*_{-s_i})] \geq
\mathbb{E}_{\mathbf{S}^*_{-s_i}}[u_i(s_i,\mathbf{S}^*_{-s_i})].
$$
\noindent At the BNE, our mechanism aims to achieve several design objectives in Section~\ref{sec:objectives}.

\subsection{Design Objectives}\label{sec:objectives}
With the strategy $\mathbf{S^*}$ at the BNE, we state three design objectives below.

\begin{definition}
(\textbf{Truthfulness}) An aggregated prediction for a query is truthful if and only if (i) the aggregation computation is correctly executed, and meanwhile, (ii) every participating provider $i\in P$ at BNE $\mathbf{S}^*$ provides a prediction $s^*_i=(\mathbf{l}_i, \mathbf{v}_i)$ satisfying the following condition:
$$
\mathbf{l}_i=\mathbf{l}^p\wedge \mathbf{D}_{KL}(\mathbf{v}^{T}||\mathbf{v}_i) \leq \theta.
$$
\end{definition}
\noindent Here, vector $\mathbf{l}^p$ contains the public possible class labels, \textit{e.g.}, \textsf{[distressed, upset, guilty, scared]} in Section~\ref{sec:combination}. $\mathbf{v}^{T}$ is the true posterior probability for $\mathbf{l}^p$.
$\mathbf{D}_{KL}(||)$ is the Kullback-Leibler (KL) divergence function.
$\mathbf{D}_{KL}(\mathbf{v}^{T}||\mathbf{v}_i)$ measures the information lost using $\mathbf{v}_i$ to approximate $\mathbf{v}^{T}$.
%
%
Clearly, condition (\emph{i}) can be guaranteed by leveraging TEEs.
Next, we design a pricing mechanism to meet condition (\emph{ii}), that is, every provider has no motivation to provide a prediction which deviates from the truthful labels and the corresponding truthful posterior probability.
However, $\mathbf{v}^{T}$ is unknown in our setting. Our designed pricing mechanism will take it into consideration.

\begin{definition}
(\textbf{Individual Rationality})
A pricing mechanism satisfies individual rationality (IR) iff every participating provider $i\in P$ at the BNE has non-negative expected utility:
$$
\mathbb{E}_{\mathbf{S}^*_{-s_i}}[u_i(s^*_i,\mathbf{S}^*_{-s_i})] \geq 0.
$$
\end{definition}

\begin{definition}
(\textbf{Budget Feasibility}) A pricing mechanism satisfies budget feasibility (BF) iff the total payment allocated to the participating providers $i\in P$ at the BNE is not more than a user's given budget for every query:
$$
\mathbb{E}_{\mathbf{S}^*}[\sum_{i=1}^mp_i(\mathbf{S}^*)]\leq \frac{B}{n},
$$
\noindent where $m$ is the number of providers while $n$ is the number of queries.
\end{definition}

\subsection{Pricing Mechanism Formulation}
We are now ready to formulate the optimization problem of designing our pricing mechanism for participants' predictions (called as PPP), \textit{i.e.},
$$
\underset{p(\cdot)}{\textbf{max}} \quad \sum_{i=1}^m\mathbf{Pr}(\mathbf{D}_{KL}(\mathbf{v}^T||\mathbf{v}_i)\leq \theta)
$$
$$
\hspace{-10pt}\text{s.t.}\hspace{5pt}\mathbb{E}_{\mathbf{S}^*_{-s_i}}[u_i(s^*_i,\mathbf{S}^*_{-s_i})] \geq 0
$$
$$
\hspace{5pt} \mathbb{E}_{\mathbf{S}^*}[\sum_{i=1}^mp_i(\mathbf{S}^*)]\leq \frac{B}{n}.
$$
As elaborated, given a set of participating providers $P=\{1,...,m\}$, $n$ queries and budget $B$, we aim to customize a payment function $p(\cdot)$ which satisfies both constraints of IR and BF, as well as maximizes the objective function, that is, the overall probability of the KL divergence between every provider's prediction at BNE $\mathbf{S^*}$ and the true prediction which is less than given threshold $\theta$.

Solving PPP optimization problem will effectively minimize the loss between the accuracy of the aggregated predictions via truth discovery and the truth accuracy, which is the user's objective.
First, given $n$ queries, $\sum_{j=1}^n\sum_{i=1}^m\mathbf{D}_{KL}(\mathbf{v}^{Tj}||\mathbf{v}_i^j)$ is apparently minimized, if PPP optimization problem is solved for every query.
Next, we can achieve that the result accuracy via truth discovery is as close as possible to the truth accuracy due to
$\sum_{j=1}^n\sum_{i=1}^m\mathbf{Pr}(\mathbf{D}_{KL}(\mathbf{v}^{Tj}||\mathbf{v}_i^j)\leq \theta) \geq \sum_{j=1}^n\mathbf{Pr}(\mathbf{D}_{KL}(\mathbf{v}^{Tj}||\mathbf{v}^j) \leq \theta)$.
The conclusion is according to the following derivation:
\vspace{-5pt}
$$
\sum_{i=1}^m\sum_{j=1}^n\mathbf{D}_{KL}(\mathbf{v}^{Tj}||\mathbf{v}_i^j)
\geq
\frac{\sum_{i=1}^mw_i(\sum_{j=1}^n\mathbf{D}_{KL}(\mathbf{v}^{Tj}||\mathbf{v}_i^j))}{\sum_{i=1}^mw_i}
$$
$$
=\sum_{j=1}^n\frac{\sum_{i=1}^mw_i\mathbf{D}_{KL}(\mathbf{v}^{Tj}||\mathbf{v}_i^j)}{\sum_{i=1}^mw_i}
\geq\sum_{j=1}^n\mathbf{D}_{KL}(\mathbf{v}^{Tj}||\frac{\sum_{i=1}^mw_i\mathbf{v}_i^j}{\sum_{i=1}^mw_i})
$$
$$
\hspace{-165pt}=\sum_{j=1}^n\mathbf{D}_{KL}(\mathbf{v}^{Tj}||\mathbf{v}^j).
$$

However, solving PPP optimization problem is hard and the ground truth is unavailable, namely $\mathbf{v}^T$.
Hence, we approximately solve it by applying the idea of divergence-based Bayesian Truth Serum (BTS) method~\cite{radanovic2014incentives}.
The main idea of the divergence-based BTS method is rewarding a player based on the divergence between her reports and a randomly selected counterpart' reports, when there is no ground truth for verification.
It is an effective approach to incentivize report truthfulness and control report quality~\cite{radanovic2014incentives}.
\wjs{We inherit such desirable properties from the divergence-based BTS method, and in the meantime, we handle both discrete data and continuous data, \emph{i.e.}, label and posterior probability, which is different from prior works~\cite{jin2017theseus, gong2018incentivizing} considering either continuous data or discrete data.}

\wjs{Derived from the divergence-based BTS method, we denote our payment function.}
It rewards a participating provider $i$ based on its strategy $s_i$ and a randomly selected provider $r$'s $s_r$ by calculating two scores.
The payment function is
$
p_{i}(s_i) = \alpha_i \times (score_{\text{I}i} + score_{\text{P}i} + 1)^2
$, where $\alpha_i > 0$.
The two scores are denoted accordingly as following:

\noindent (1) $score_{\text{I}i} = score_{\text{I}}(s_i, s_r)$ measures a \textit{penalty} value if $s_i$ reports the same labels with $s_r$, but the corresponding posterior probability disagrees with each others.
    $$
    score_{\text{I}}(s_i, s_r) = -\mathbb{I}_{\mathbf{l}_i=\mathbf{l}_r\wedge \mathbf{D}_{KL}(\mathbf{v}_i||\mathbf{v}_r)>\theta}
    $$
Herein, $\mathbb{I}_{a}$ is an indicator. Its value is $1$, if condition $a$ is valid; otherwise, its value is $0$.

\noindent (2)  $score_{\text{P}i} = score_{\text{P}}(s_i, s_r)$ measures a \textit{reward} value if $s_i$'s posterior probability fits close to the distribution of the class labels provided by $s_r$.
\begin{align*}
score_{\text{P}}(s_i, s_r)=\frac{1}{c}\sum_{k=1}^c[2 - (1-\mathbf{v}_{i}(l_{rk}))^2
-\sum_{l_r\in \Omega/\{l_{rk}\}} \mathbf{v}_{i}(l_r)^2]
\end{align*}
Herein, $c$ is the number of possible class labels;
$\mathbf{v}_i(l)$ means the posterior probability for label $l$ and
$(\mathbf{v}_i(l_{1}), ..., \mathbf{v}_i(l_{c}))$ constitute $\mathbf{v}_i$ with constraint $\sum_{k=1}^c\mathbf{v}_i(l_{k})=1$.
Concretely, $\mathbf{v}_{i}(l_{rk})=\mathbf{Pr}(l_{rk}|l_{ik})$ represents $i$'s the posterior probability for label $l_{rk} \in \Omega$.
If $\mathbf{v}_{i}(l_{rk})=\mathbf{Pr}(l_{rk}|l_{ik}=l_{rk})$, the score value is maximized, being equal to $2$.
If $\mathbf{v}_{i}(l_{rk})=\mathbf{Pr}(l_{rk}|l_{ik}\neq l_{rk})$, the score value is minimized, being equal to $0$.\looseness=-1

With the definitions above, the value of $(score_{\text{I}i} + score_{\text{P}i} + 1)$ falls in the range $[0,3]$.
Also, it is worth noting that the scoring rule consisting of the two scores has been proved \textit{strictly Bayes-Nash incentive-compatible} relying on  \textit{stochastic relevance} in~\cite{radanovic2014incentives}.
It means that a truthful prediction is always configured with a higher score compared to a untruthful prediction so as to achieve the goal of fairness (refer it to Section~\ref{sec:goals}).

Considering the potentiality of a participating provider aborting, we revise our payment function. If provider $i$ does not abort, her deposit $d_0$ should be refunded, that is, $p_i(s_i)=p_i(s_i) + d_0$. Otherwise, her deposit $d_0$ will be forfeited.

\subsection{Analysis}
In this section, we proceed to analyze how to achieve the design objectives in Section~\ref{sec:objectives} by using the presented pricing function as an approximately solution.

To begin with, we quantify the cost function with respect to different participating providers, which is useful to estimate the providers' expected utility.
For simplicity, we assume that participating providers' costs are known, which refers to the \emph{complete information} scenario. Their costs derive from the identical two cost parameters $c_1>0$ and $c_2>0$ which are far smaller than a user's budget $B$.
We assume that the cost of generating a product linearly increases with the product's quality.
Recall that we measure the truthfulness of a prediction via the divergence-based BTS method due to the lack of ground truths.
Specifically, using two scores measures a prediction truthfulness.
Thus, we next naturally regard the two scores as the quality metric to calculate the corresponding cost of every strategy $s_i$. That is, $c(s_i)=c_1\cdot (score_{\text{I}i}+score_{\text{P}i}+1)+c_2$. It is noteworthy that the cost monotonically increases with score $(score_{\text{I}i}+score_{\text{P}i}+1)$ increasing.

We are now ready to analyze that with our pricing mechanism, there exists a BNE achieving our design objectives via parameter constraints.
Specifically, we set constraint conditions on parameter $\alpha_i$ considering the design objectives of individual rationality and budget feasibility, based on which we find a BNE, where all participants adopt the strategy of offering truthful predictions.
Below, we demonstrate and prove this finding by Theorem~\ref{the:pricing}.
\begin{theorem}\label{the:pricing}
In the non-cooperative game, there exists a BNE $\mathbf{S}^*=(s_1^*, ..., s_m^*)$, where every participating provider $i\in\{1,...,m\}$ provides $s_i^*$ containing $\mathbf{l}_i=\mathbf{l}_r$ and $\mathbf{D}_{KL}(\mathbf{v}_i||\mathbf{v}_r)\leq \theta$ compared with $s_r^*$ $(r\neq i)$ when parameter $\alpha_i$ satisfies (1) $\alpha_i \geq \frac{c_1}{2\cdot score_i}$, (2)
$\alpha_i \geq \frac{c_1\cdot score_i + c_2}{score_i^2}$ and
(3) $\alpha_i \leq \frac{B}{nm}\cdot\frac{1}{score_i^2}$, where $score_i = (score_{\text{I}i}+score_{\text{P}i}+1)$.
\end{theorem}

\begin{proof}
Given other participating providers' strategies $\mathbf{S}_{-s_i}^*$ and a randomly selected provider's strategy $s_r^*$, every provider $i$ can estimate her expected utility by
\begin{align*}
\mathbb{E}_{\mathbf{S}_{-s_i}}[u_i(s_i,\mathbf{S}_{-s_i})|s_r^*] &=
\mathbb{E}_{\mathbf{S}_{-s_i}}[p_i(s_i, \mathbf{S}_{-s_i})|s_r^*] - c(s_i) - d_0\\
&=\alpha_i \times (score_{\text{I}i} + score_{\text{P}i} + 1)^2 + d_0 \\
& \quad - c_1\cdot (score_{\text{I}i}+score_{\text{P}i}+1)\\
& \quad - c_2 - d_0.
\end{align*}
Here, we suppose that provider $i$ does not abort. If she aborts, apparently her expected utility is equal to $-d_0$ which is negative.
For every provider $i$ not aborting, she can maximize her expected utility when her strategy $s_i^*=(\mathbf{l}_i,\mathbf{v}_i)$ leads to $score_i$ reaching the maximum among $[\frac{c_1}{2\alpha_i}, 3]$.
Therefore, every rational provider $i$ is doomed to chose the strategy which enables $score_i$ being equal to $3$.
To be more clear, we summarize the possible cases for every provider $i$'s strategy and her expected utility as following:

\noindent(a). If $s_i^*=(\mathbf{l}_i,\mathbf{v}_i)$, where $\mathbf{l}_i \neq \mathbf{l}_r$, her expected utility is negative due to
$(score_{\text{I}i} + score_{\text{P}i} + 1)=0$ leading to $\mathbb{E}_{\mathbf{S}_{-s_i}}[u_i(s_i,\mathbf{S}_{-s_i})|s_r^*]=-c_2$.

\noindent(b). If $s_i^*=(\mathbf{l}_i,\mathbf{v}_i)$, where $\mathbf{l}_i=\mathbf{l}_r$ and $\mathbf{D}_{KL}(\mathbf{v}_i||\mathbf{v}_r)\leq \theta$, her expected utility is equal to $9\alpha_i-3c_1-c_2$ which is positive due to parameter constraint (2), and maximized due to $(score_{\text{I}i} + score_{\text{P}i} + 1)=3$.

\noindent(c). If abort, her expected utility is negative due to $\mathbb{E}_{\mathbf{S}_{-s_i}}[u_i(s_i,\mathbf{S}_{-s_i})|s_r^*]=-d_0$.

\noindent Hence, strategy profile $\mathbf{S}^*=(s_1^*, ..., s_m^*)$ in Theorem~\ref{the:pricing}, where $s_i^*$ satisfies $\mathbf{l}_i=\mathbf{l}_r$ and $\mathbf{D}_{KL}(\mathbf{v}_i||\mathbf{v}_r)\leq \theta$ is a BNE.
\end{proof}
\section{Experiment}\label{sec:exp}

\subsection{Implementation and Setup}
\noindent\textbf{Prediction Aggregation with TEEs.}
We initialize TEEs by utilizing SGX SDK of version 2.5.
In the SGX environment, we implement the prediction aggregation program (\textit{i.e.}, Algorithm~\ref{alg:truthdiscovery}) by using C/C++ programming language.

\noindent\textbf{Smart Contract.} We also implement the PS contract with the Solidity programming language of Ethereum and deploy it on the Ropsten Test Network via MetaMask\footnote{https://metamask.io/}.
%

%
\noindent\textbf{Dataset.}
We totally use three datasets to simulate three prediction tasks.
Specifically, we use two well-studied image datasets, including MNIST\footnote{http://yann.lecun.com/exdb/mnist/} and ImageNet\footnote{http://www.image-net.org/challenges/LSVRC/2012/} for image prediction, as well as a public text dataset, namely 20 Newsgroups\footnote{http://qwone.com/~jason/20Newsgroups/} for text prediction.
%
%
%
With respect to three datasets, we will correspondingly sample a number of test data for evaluation.
Note that MNIST, ImageNet and 20 Newsgroups contain $10$K, $100$K and near $8$K test data, respectively.
More concrete information of the three datasets are shown in TABLE~\ref{tab:datasets}.
%

\begin{table}[htbp]
 \centering
 \caption{\label{tab:datasets} Real-world datasets used in the experiment.}
 \begin{tabular}{lcccc}
  \toprule
  \textbf{Dataset} & \textbf{Type} & \textbf{Size} & \textbf{Features} & \textbf{Labels} \\
  \midrule
  MNIST & Image & $70$K   & $20$x$20$  & 10\\
  ImageNet & Image & $1.26$M  & $224$x$224$x$3$  & 1000\\
  20 Newsgroups & Text & $18846$  & --  &20 \\
  \bottomrule
 \end{tabular}
 \vspace{-5pt}
\end{table}

\noindent\textbf{Provider Simulation.}
We collect three groups of various trained models which are used to simulate providers for serving prediction.
We separately collect $6$, $10$ and $15$ models under various frameworks which are evaluated on MNIST, 20 Newsgroups and ImageNet.
Specifically, we implement and train the models for the MNIST and 20 Newsgroups by ourselves, and download off-the-shelf models for ImageNet from two public model sources\footnote{https://keras.io/api/applications/}$^{,}$\footnote{https://pytorch.org/docs/stable/torchvision/models.html}.
%
%
Due to the space limitation, we only present the models trained on 20 Newsgroups dataset in TABLE~\ref{tab:models-N}.
%
%
\begin{table}[htbp]
  \centering
  \caption{\label{tab:models-N} Models evaluated on 20 Newsgroups.}
    \begin{tabular}{lcl|lcl}
    \toprule

    \textbf{Model} & \textbf{Framework} & \textbf{Acc.} & \textbf{Model} & \textbf{Framework} & \textbf{Acc.} \\
    \midrule
    Boost & SKLearn & 0.740  & KNN   & SKLearn & 0.660 \\
    Bagging & SKLearn & 0.660  & CNN   & Keras & 0.730 \\
    Dec. tree & SKLearn & 0.550  & DNN   & Keras & 0.810 \\
    Ran. forest & SKLearn & 0.760  & RNN   & Keras & 0.760 \\
    SVM   & SKLearn & 0.820  & RCNN  & Keras & 0.720 \\
    \bottomrule
    \end{tabular}%
  \label{tab:addlabel}%
\vspace{-10pt}
\end{table}%

We simulate distrustful predictions by perturbing normal predictions, where perturbations are sampled from the uniform distribution on interval $(0, 1)$.
%
%
With the perturbation methods, we simulate a distrusting provider by perturbing a model's all predictions.

We will consider three cases in perturbing predictions of models, including (a) no perturbation, (b) perturbing no more than $\frac{M}{2}$ models' predictions, and (c) perturbing more than $\frac{M}{2}$ models' predictions, where $M$ is the total number of models.
Note that case (a) is used to simulate the BNE setting induced by Theorem~\ref{the:pricing}, where each provider is incentivized to provide truthful predictions; case (c) creates the setting, where providers lack sufficient motivation for prediction truthfulness; case (b) refers to the setting between case (a) and (c).

In addition, our experiments are conducted in a Ubuntu 16.04 server equipped with a CPU of 3.40GHz, 32 GB RAM and a GPU of Nvidia GTX-1080.\looseness=-1

\subsection{Evaluation}
Our evaluation is four-fold:
(\emph{i}) To highlight the advantage of Algorithm~\ref{alg:truthdiscovery}, we compare the accuracy of predictions generated by Algorithm~\ref{alg:truthdiscovery} and that by averaging (a traditional ensemble strategy);
(\emph{ii}) To demonstrate the effectiveness of the incentive mechanism, we plot and compare simulation results of prediction aggregation regarding case (a), (b) and (c) in terms of accuracy;
(\emph{iii}) To show service cost, we estimate the computation complexity of prediction aggregation with a TEE and evaluate gas costs caused by the interaction between the PS contract and the TEE;
(\emph{iv}) To answer whether or not prediction aggregation via Algorithm~\ref{alg:truthdiscovery} is effective to resist membership inference attacks, we conduct state-of-the-art attacks~\cite{salem2019ml} and present empirical evidences.\looseness=-1

\begin{table}[htbp]
\vspace{-5pt}
 \centering
 \caption{\label{tab:accuracy} Accuracy comparison.}
 \begin{tabular}{lcccc}
  \toprule
  \textbf{Dataset} & \textbf{Avg.} & \textbf{Label} & \textbf{Rank} & \textbf{Probability} \\
  \midrule
  MNIST & $0.907$ & $0.978$   & $0.973$  & $0.981$\\
  ImageNet & $0.724$ & $0.790$  & $0.764$  & $0.789$\\
  20 Newsgroups & $0.721$ & $0.862$  & $0.836$  & $0.862$\\
  \bottomrule
 \end{tabular}
 \vspace{-5pt}
\end{table}
First of all, as shown in TABLE~\ref{tab:accuracy}, for each dataset, the accuracy of the predictions generated by Algorithm~\ref{alg:truthdiscovery} regrading three output formats (\emph{i.e.}, $3_{th}$ to $5_{th}$ column) is always better than the averaging accuracy (\emph{i.e.}, $2_{th}$ column) of all participating models.
We can see that on ImageNet dataset, the accuracy improvement is relatively small, but as pointed out by~\cite{russakovsky2015imagenet}, spending a lot of time and energy to achieve minor accuracy improvement on difficult object recognition task is deserved.\looseness=-1

Second, Fig.~\ref{fig:case} and Fig.~\ref{fig:format} (from left to right) show the accuracy of aggregated predictions regarding three perturbation cases on MNIST, 20 Newsgroups and ImageNet, respectively.
For each dataset, it can be clearly seen that the accuracy in case (a) is always higher than that in case (b) and (c), which is because that participating providers offer truthful predictions with sufficient incentives.
We also can see that in case (c), where a vast majority of participating providers report meaningless predictions, the accuracy is never better than $0.5$.
The reason is that Algorithm~\ref{alg:truthdiscovery} fails to learn the truth when a majority of predictions are not enough accurate, and thus our incentive mechanism is necessary to handle case (c).
In addition, from Fig.~\ref{fig:case}, the evaluated accuracy slightly grows up with the increasing queries.
According to Fig.~\ref{fig:format}, we also notice that the accuracy of the rank-level predictions on ImageNet drops more obviously than the other two datasets in more serious perturbation cases.
It might be caused by the large number of labels, \emph{i.e.}, 1000, on ImageNet dataset.

Third, Fig.~\ref{fig:time} presents the estimated time costs of prediction aggregation inside the TEE over three datasets.
Note that we omit the one-time cost of setting up a TEE.
Clearly, more queries spend more times.
By comparing the three sub-figures, we also can know that the time complexity becomes higher as the number of labels of the query task increases.
Recall that the number of class labels of MNIST, 20 Newsgroups and ImageNet is $10$, $20$ and $1000$, respectively.
Besides, gas costs are mainly derived from two parts:
(1) execution costs of the PS contract when its three entry points, \textsf{Deposit}, \textsf{Request} and \textsf{Response}, are correspondingly invoked, and
(2) execution costs of the TEE's transaction (on entry point \textsf{Response}) which contains outputs $outp$, $outp_{attr}$  and signatures $\sigma$, $\sigma_{attr}$ (entirely $2\times 70$ bytes).
Also, the gas costs grow up with increasing participating providers.
Note that encrypted input data and predictions are transmitted off-chain, and thus the magnitude of query makes negligible effect on the gas costs.
We only test the gas costs by simulating 6 providers (on MNIST).
Specifically, part (1) totally spends $510,815$~units gas, including $389,373$~units for \textsf{Deposit}, $102,200 $~units for \textsf{Request} and $19,242$~units for \textsf{Response}.
The gas costs for sending the response transaction in part (2) are about $74,370$~units.\looseness=-1

\begin{figure}
	\centering
   \hspace{-20pt}
	\subfigure{
		\begin{minipage}[b]{0.13\textwidth}
			\includegraphics[width=1.3\textwidth]{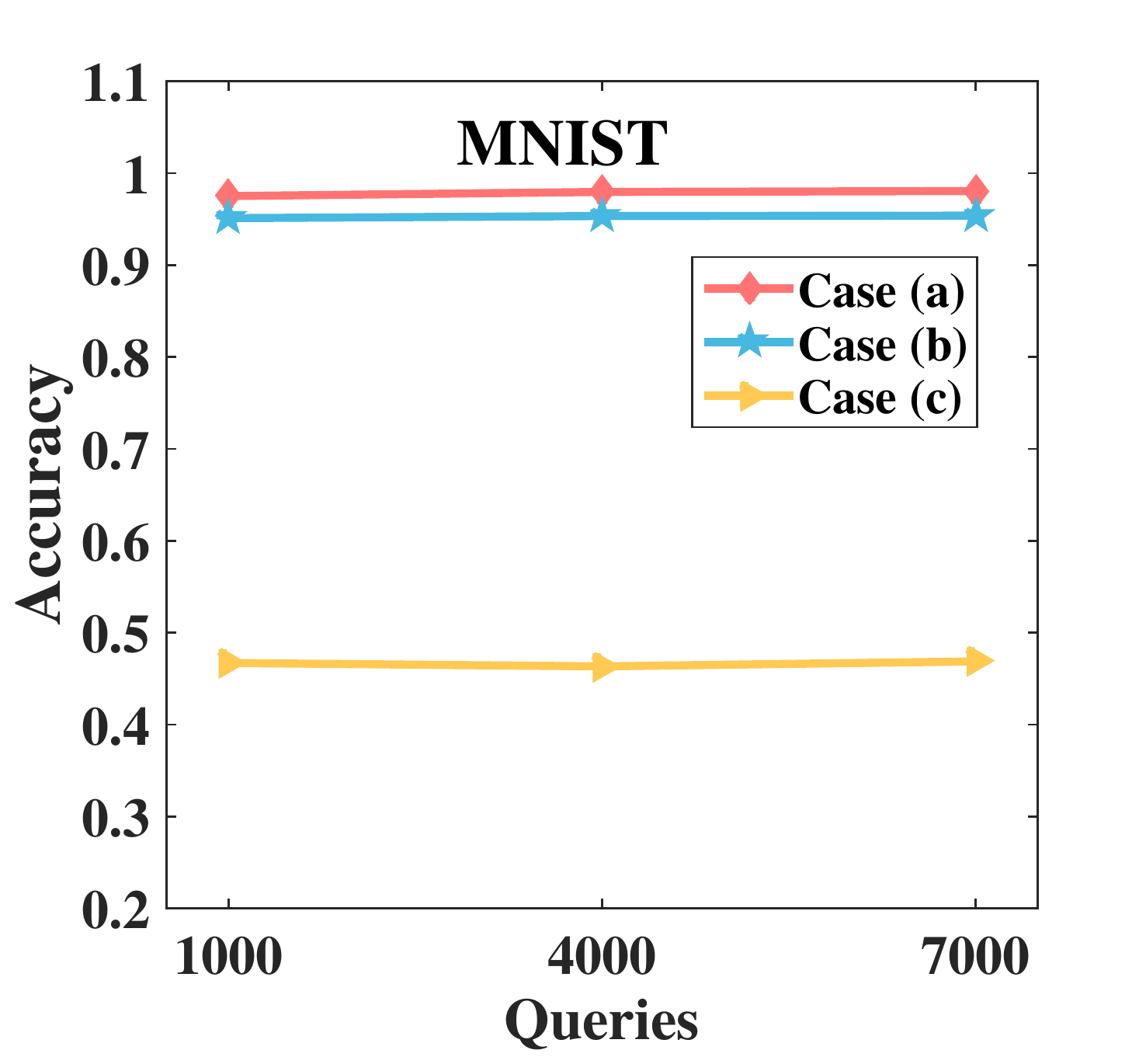}
		\end{minipage}
			\label{fig:case:mnist}
	}
	\hspace{10pt}
    	\subfigure{
    		\begin{minipage}[b]{0.13\textwidth}
   		 	\includegraphics[width=1.3\textwidth]{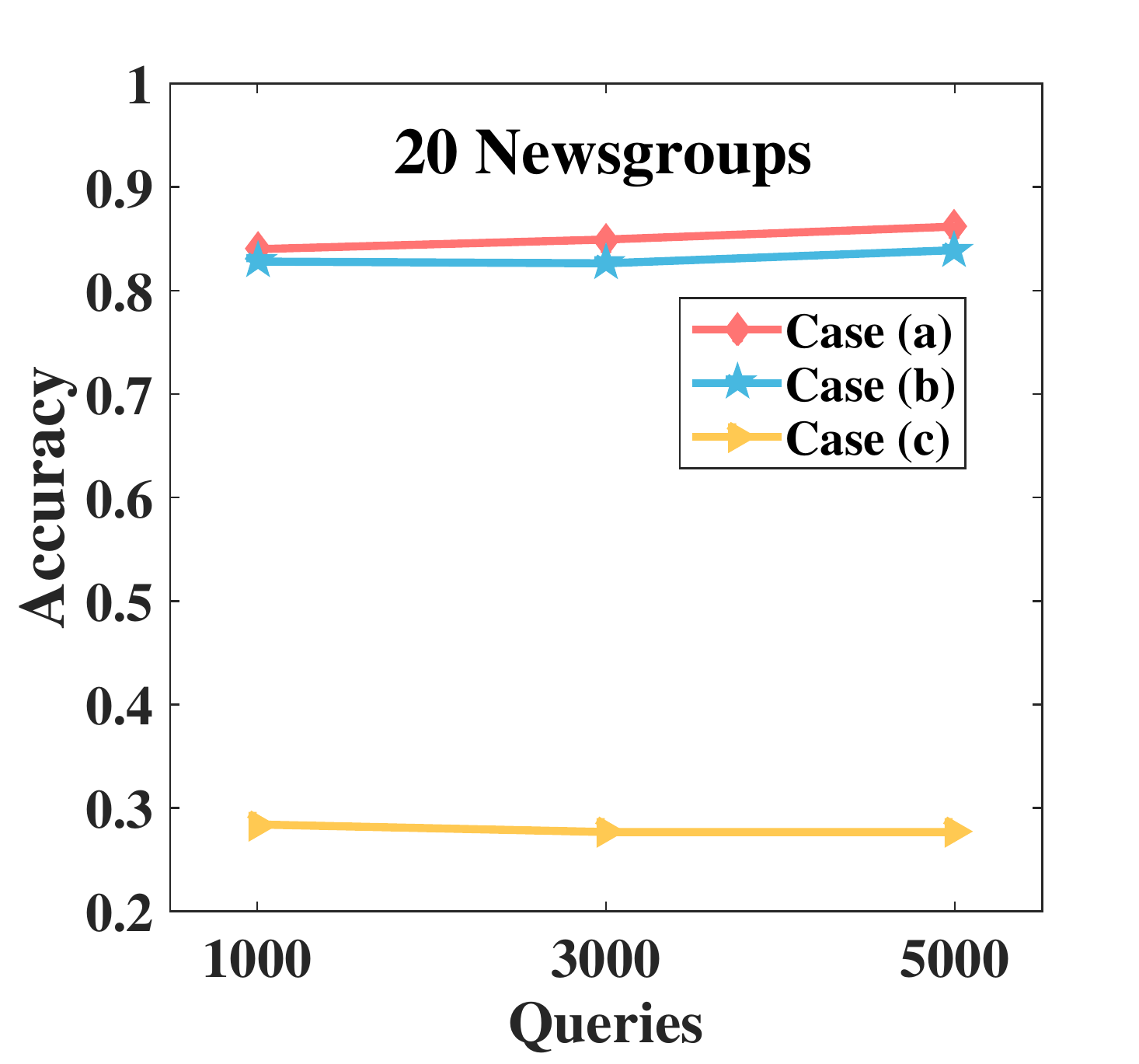}
    		\end{minipage}
		\label{fig:case:news}
    	}
    \hspace{10pt}
    	\subfigure{
    		\begin{minipage}[b]{0.13\textwidth}
   		 	\includegraphics[width=1.3\textwidth]{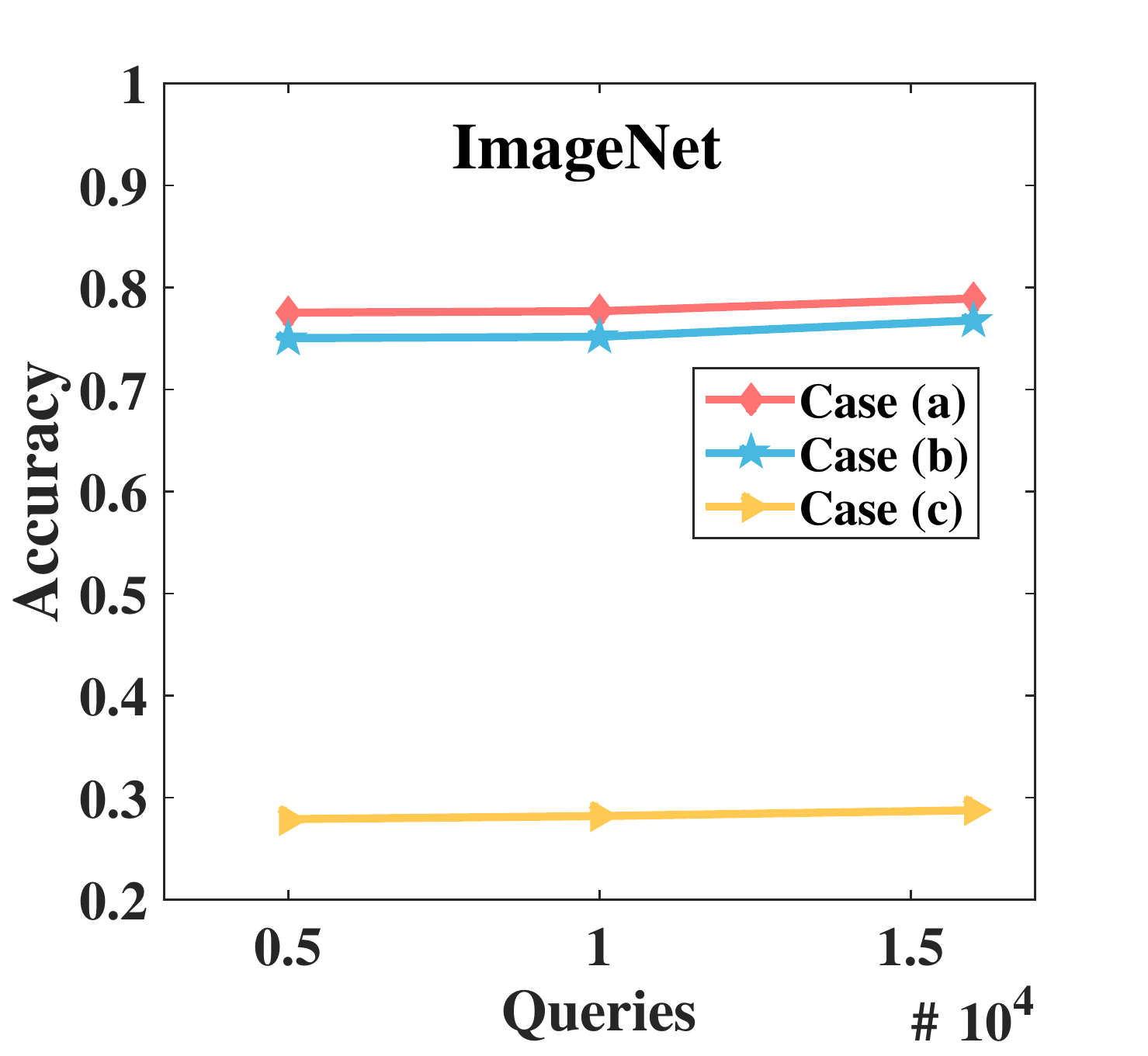}
    		\end{minipage}
		\label{fig:case:image}
    	}
	\caption{Accuracy with increasing queries..}
	\label{fig:case}
	\vspace{-10pt}
\end{figure}

\begin{figure}

	\centering
   \hspace{-20pt}
	\subfigure{
		\begin{minipage}[b]{0.13\textwidth}
			\includegraphics[width=1.3\textwidth]{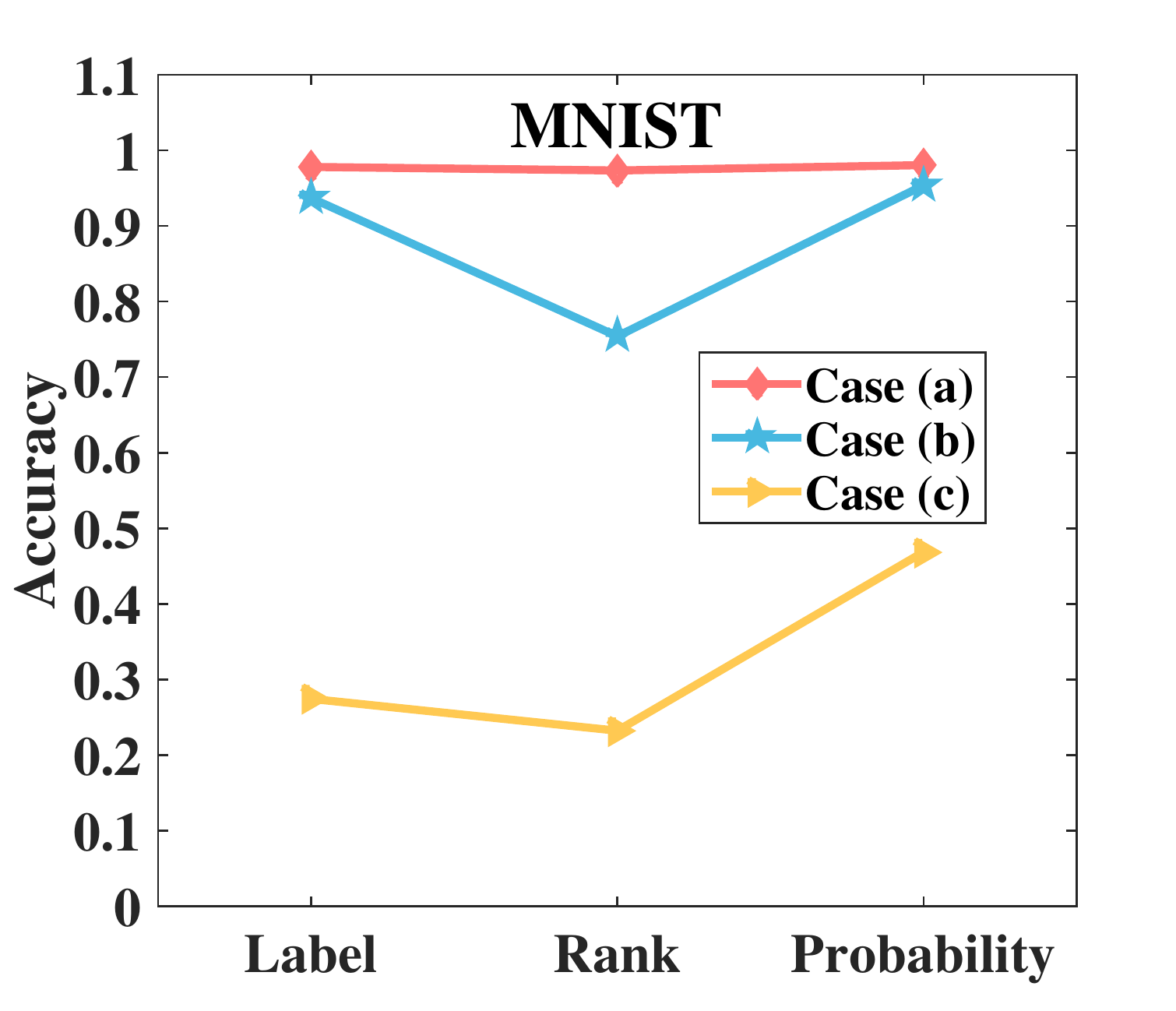}
		\end{minipage}
			\label{fig:format:mnist}
	}
	\hspace{10pt}
    	\subfigure{
    		\begin{minipage}[b]{0.13\textwidth}
   		 	\includegraphics[width=1.3\textwidth]{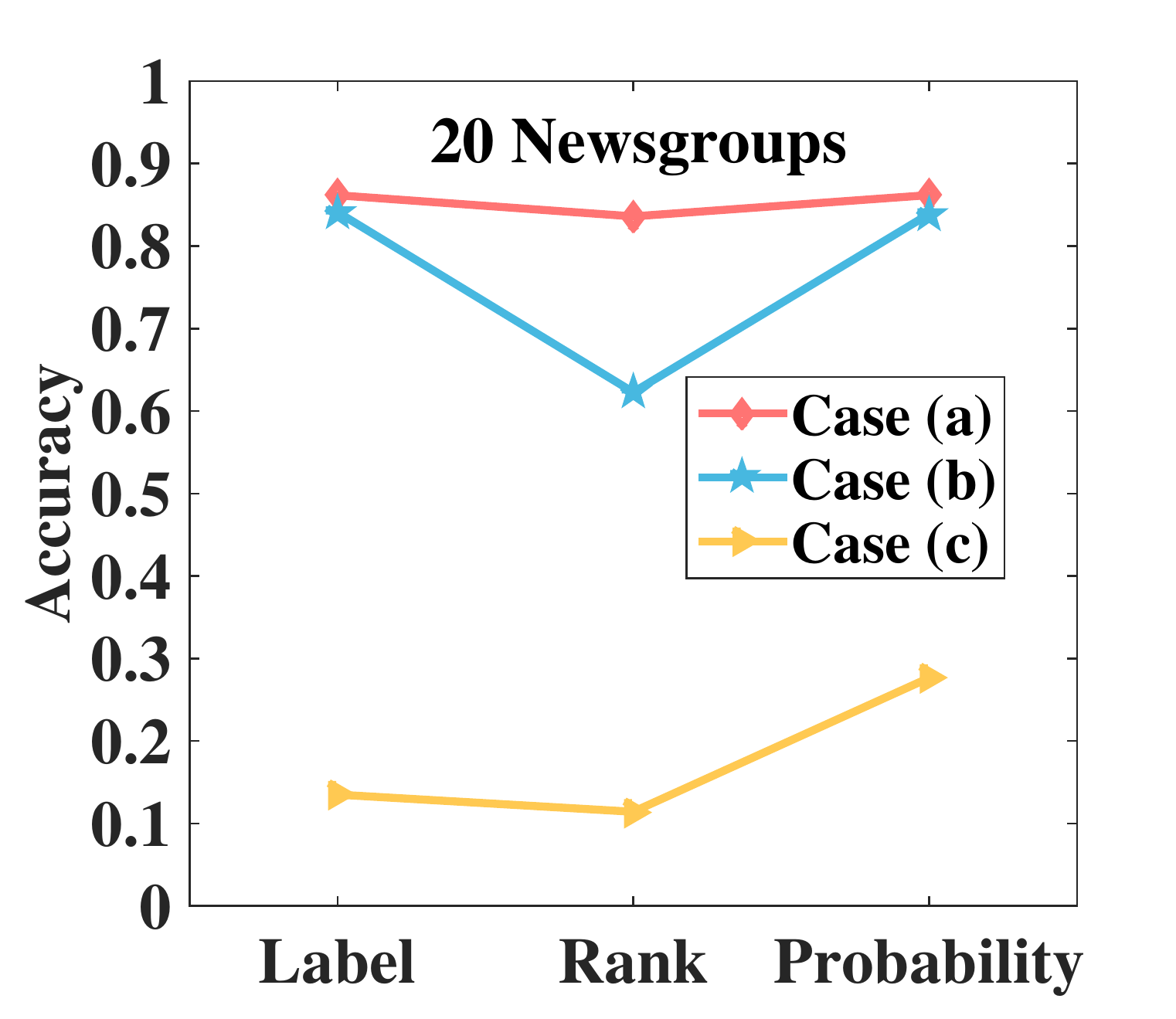}
    		\end{minipage}
		\label{fig:format:news}
    	}
    \hspace{10pt}
    	\subfigure{
    		\begin{minipage}[b]{0.13\textwidth}
   		 	\includegraphics[width=1.3\textwidth]{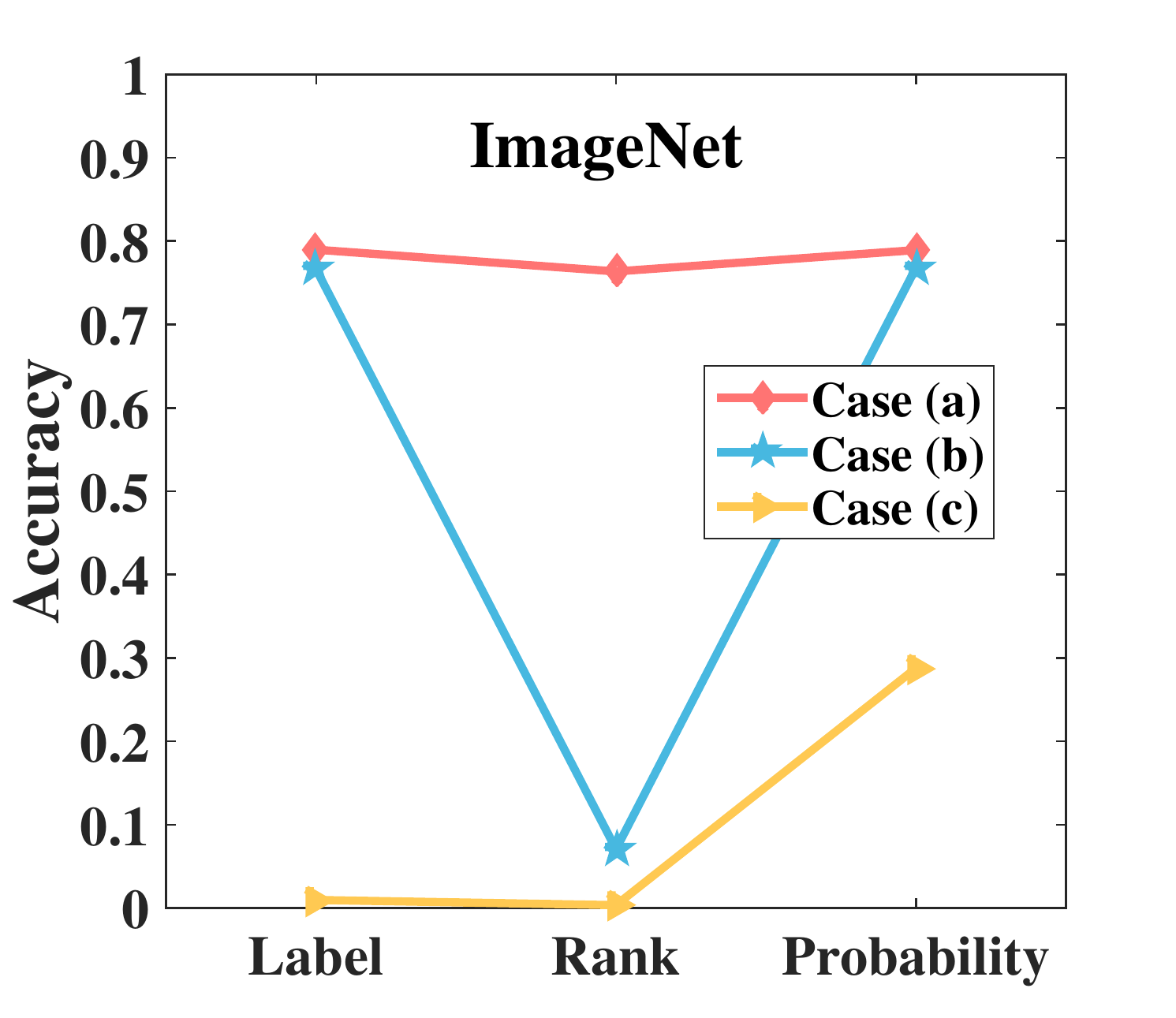}
    		\end{minipage}
		\label{fig:format:image}
    	}
	\caption{Accuracy on different-format predictions.}
	\label{fig:format}
	\vspace{-20pt}
\end{figure}

\begin{table}[htbp]
\vspace{-5pt}
  \centering
  \caption{Comparison of adversaries' attack performance.}
    \begin{tabular}{cccc}
    \toprule
    \textbf{Type} & \textbf{Target model} & \textbf{Precision} & \textbf{Recall} \\
    \midrule
    \multirow{2}[4]{*}{\textbf{Adversary 1}} & Single & 0.996 &  0.503\\
\cmidrule{2-4}  &  Ensemble &  \textbf{0.056}  & \textbf{0.054} \\
    \midrule
    \multirow{2}[4]{*}{\textbf{Adversary 2}} & Single &  0.997 & 0.504\\
\cmidrule{2-4}  & Ensemble & 0.987 & 0.499 \\
    \bottomrule
    \end{tabular}%
  \label{tab:mia}%
  \vspace{-10pt}
\end{table}%

Last, we launch membership inference attacks using two types of adversaries with increasingly strong attack capabilities in prior work (\emph{i.e.}, adversary 1 and 2, detailed in \cite{salem2019ml}'s TABLE I) and show the attack results.
Similar to the work~\cite{salem2019ml}, we adopt three models as an ensemble, but the difference is that our ensemble strategy is Algorithm~\ref{alg:truthdiscovery} rather than stacking.
Besides, the used three models are CNN, RNN and MLP trained on the MNIST dataset.
For comparison, we also conduct the same attacks on the single CNN model.
As shown in TABLE~\ref{tab:mia}, the attack results demonstrate that ensemble model under Algorithm~\ref{alg:truthdiscovery} is able to reduce the attack performance of adversary 1, but not adversary 2.
Concretely, for adversary 1, the precision drops from $0.996$ to $0.056$ and the recall drops from $0.503$ to $0.054$.
But for adversary 2, there has no effect.
It is difficult to suggest certain confident explanation for the attack results like the previous work~\cite{salem2019ml}.\looseness=-1

\begin{figure}
	\centering
   \hspace{-20pt}
	\subfigure{
		\begin{minipage}{0.13\textwidth}
			\includegraphics[width=1.3\textwidth]{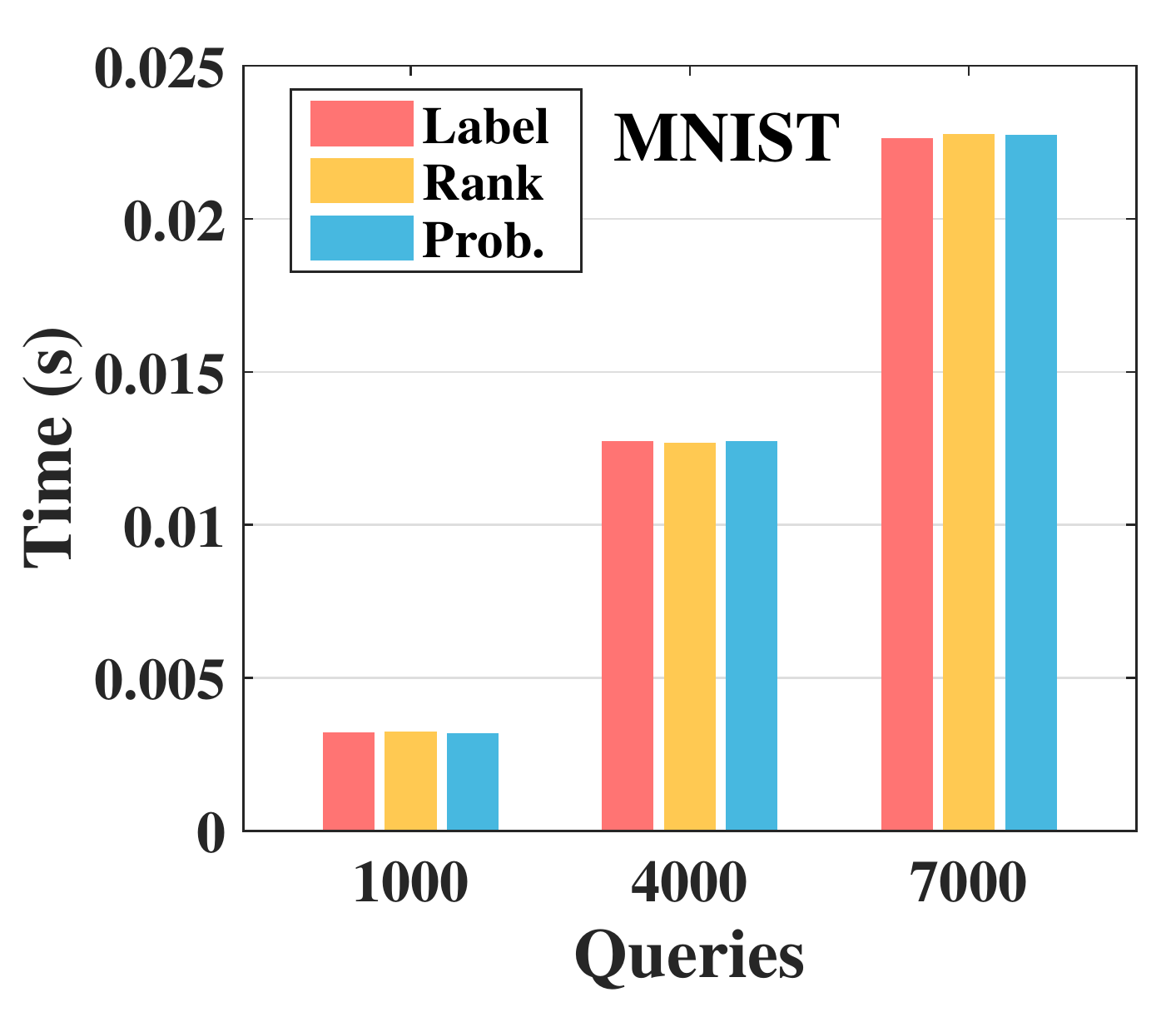}
		\end{minipage}
			\label{fig:time:mnist}
	}
	\hspace{10pt}
    	\subfigure{
    		\begin{minipage}{0.13\textwidth}
   		 	\includegraphics[width=1.3\textwidth]{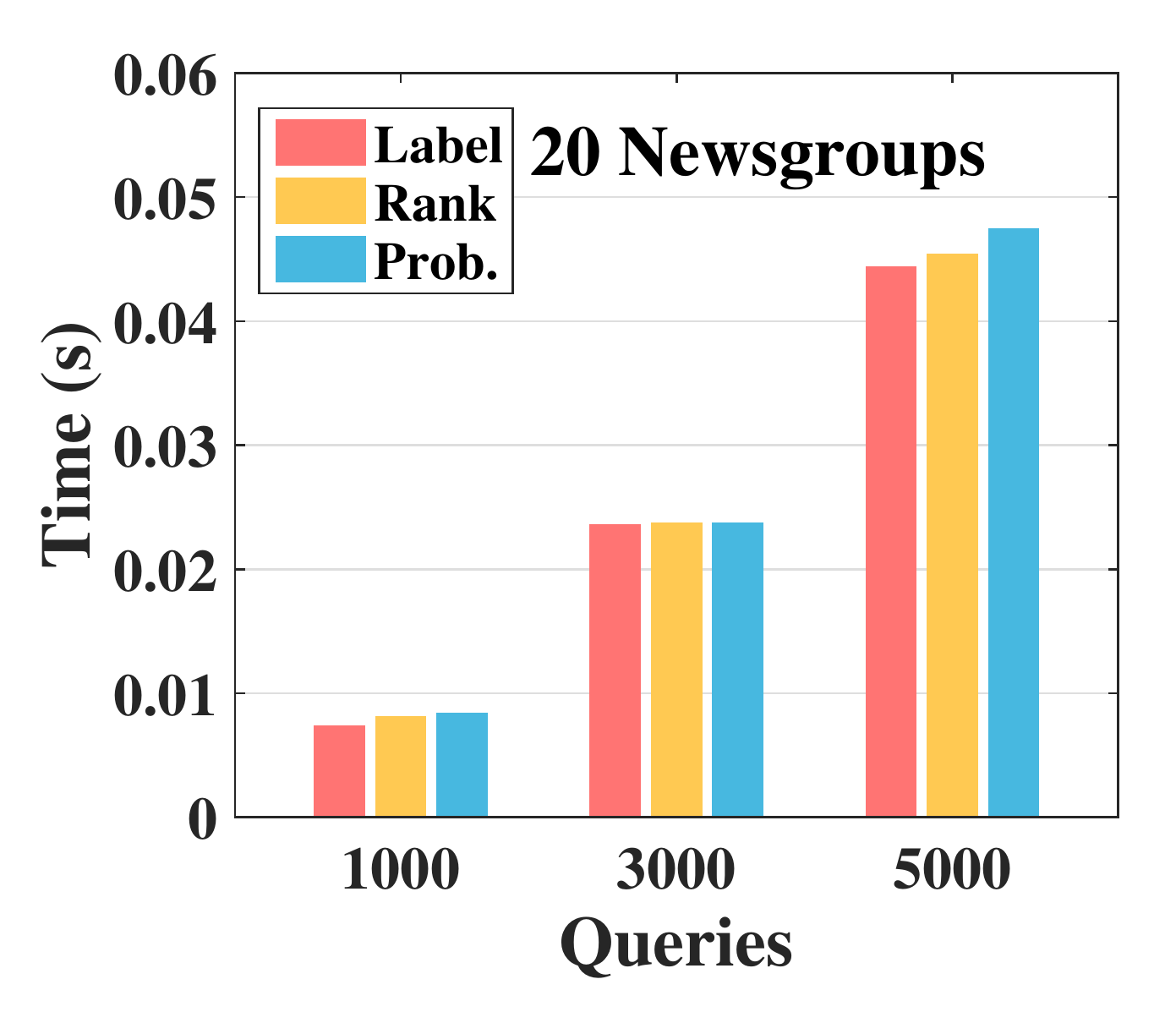}
    		\end{minipage}
		\label{fig:time:news}
    	}
    \hspace{10pt}
    	\subfigure{
    		\begin{minipage}{0.13\textwidth}
   		 	\includegraphics[width=1.3\textwidth]{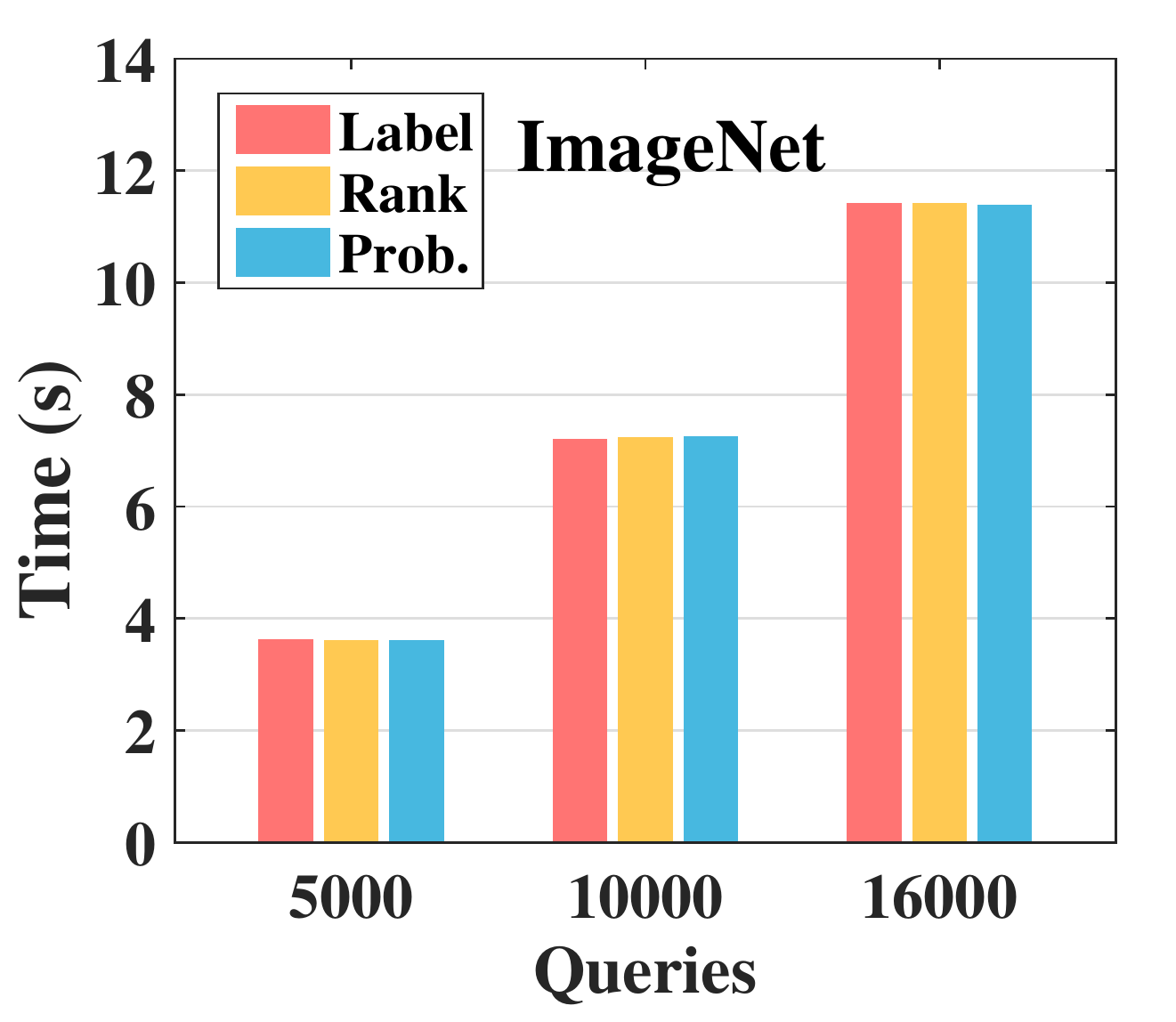}
    		\end{minipage}
		\label{fig:time:image}
    	}
	\caption{Time complexity of prediction aggregation inside a TEE.}\label{fig:time}
\vspace{-20pt}
\end{figure}

\section{Limitation and Future Work}\label{sec:limitation}
\noindent\textbf{Other prediction formats.}
Our work focuses on the prediction formats, including top-1 label, ranked labels and posterior probability, but fails to support other formats, such as text data, in Natural language processing (NLP) tasks.
Taking language translation as an example, Sequence-to-Sequence models are usually used, which take as input a sequence of words in certain language and output another sequence of words in a target language, where output format belongs to text data.\looseness=-1
%

\noindent\textbf{Adversarial examples.}
We assume benign input data and do not consider adversarial examples (AEs), \emph{i.e.}, input data injected with imperceptible perturbations~\cite{szegedy2013intriguing}.
AEs can mislead a deep neural network to incorrectly classify
an originally correctly classified input.
Recently, a promising approach against AEs is to create a robust ensemble model by carefully considering the diversity of individual models~\cite{pang2019improving, liu2019deep}.
In our future work, we will follow this direction and take into account the factors regarding model diversity to refine our incentive mechanism for FedServing.

\section{Conclusion}\label{sec:conclusion}
In this paper, we present a prediction serving framework, named as FedServing, towards trained models from various sources.
FedServing enables locally deploying models and provides collective prediction services for charging users.
For motivating truthful predictions, we customize an incentive mechanism based on Bayesian game theory.
For boosting prediction accuracy, we use truth discovery algorithms working jointly with the incentive mechanism to eliminate the effect of low-accuracy predictions.
Our proposed design supports popular prediction formats, including top-1 label, ranked labels and posterior probability.
Besides, we build FedServing on the blockchain to ensure exchange fairness and leverage TEEs to securely aggregate predictions as well.
With extensive experiments, we effectively validate the expected properties of our mechanism and empirically demonstrate its capability of reducing the risk of certain membership inference attack.

\bibliographystyle{IEEEtran}
\balance
\small
\bibliography{reference}

\end{document}